\documentclass[12pt]{article}
\usepackage[margin=1in]{geometry}

\usepackage[utf8]{inputenc} 
\usepackage[T1]{fontenc}    
\usepackage{hyperref}       
\usepackage{url}            
\usepackage{booktabs}       
\usepackage{amsfonts}       
\usepackage{nicefrac}       
\usepackage{microtype}      

\usepackage{amsmath}

\usepackage{graphicx}
\usepackage{amsmath}
\usepackage{amsfonts}
\usepackage{amsthm}
\usepackage{hyperref}
\usepackage{bm}

\usepackage{algorithm}
\usepackage{algorithmic}
\usepackage{subcaption}
\usepackage{authblk}

\newtheorem{theorem}{Theorem}[section]
\newtheorem{lemma}[theorem]{Lemma}

\newtheorem{claim}[theorem]{Claim}

\newcommand{\MAB}{\textsf{MAB}}

\renewcommand{\vec}[1]{\bm{#1}}

\DeclareMathOperator*{\argmax}{arg\,max}

\begin{document}
\title{Stochastic Bandits for Multi-platform Budget Optimization in Online Advertising}

\author[1]{Vashist Avadhanula\thanks{vas1089@fb.com}}
\author[1]{Riccardo Colini-Baldeschi\thanks{rickuz@fb.com}}
\author[2]{Stefano Leonardi\thanks{Supported by ERC Advanced Grant 788893 AMDROMA ``Algorithmic and Mechanism Design Research in Online Markets'' and MIUR PRIN project ALGADIMAR ``Algorithms, Games, and Digital Markets''. leonardi@diag.uniroma1.it}}
\author[3]{Karthik Abinav Sankararaman\thanks{karthikabinavs@fb.com}}
\author[1]{Okke Schrijvers\thanks{okke@fb.com}}
\affil[1]{Facebook Core Data Science}
\affil[2]{Sapienza University of Rome \& Facebook Core Data Science}
\affil[3]{Facebook AI Applied Research}

\maketitle
\begin{abstract}
We study the problem of an online advertising system that wants to optimally spend an advertiser's given budget for a campaign across multiple platforms, without knowing the value for showing an ad to the users on those platforms.  We model this challenging practical application as a Stochastic Bandits with Knapsacks problem over $T$ rounds of bidding with the set of arms given by the set of distinct bidding $m$-tuples, where $m$ is the number of platforms. We modify the algorithm proposed in Badanidiyuru \emph{et al.,} \cite{Badanidiyuru:2018} to extend it to the case of multiple platforms to obtain an algorithm for both the discrete and continuous bid-spaces. Namely, for discrete bid spaces we give an algorithm with regret $O\left(OPT \sqrt {\frac{mn}{B} }+ \sqrt{mn OPT}\right)$, where $OPT$ is the performance of the optimal algorithm that knows the distributions. For continuous bid spaces the regret of our algorithm is  $\tilde{O}\left(m^{1/3} \cdot \min\left\{ B^{2/3}, (m T)^{2/3} \right\} \right)$. When restricted to this special-case, this bound improves over Sankararaman and Slivkins \cite{Karthik-aistats18} in the regime $OPT \ll T$, as is the case in the particular application at hand. Second, we show an $ \Omega\left (\sqrt {m OPT} \right)$ lower bound for the discrete case and an  $\Omega\left( m^{1/3} B^{2/3}\right)$ lower bound for the continuous setting, almost matching the upper bounds. Finally, we use a real-world data set from a large internet online advertising company with multiple ad platforms and show that our algorithms outperform common benchmarks and satisfy the required properties warranted in the real-world application.
\end{abstract}
\section{Introduction}
\label{s:intro}

As online advertising has proliferated, ad campaigns have moved from ad-hoc bidding for individual users, to campaigns that try to reach massive audiences while respecting marketing budgets. In recent years this has led online advertising marketplaces to offer budget management solutions that bid on advertiser's behalf to optimally spend a given budget. But, as the online advertising ecosystem grows in size and complexity, those budget management solutions have to take an increasing number of factors into account when optimizing an advertiser's spend.

One such problem that is becoming increasingly ubiquitous in online advertising is designing a bidding strategy to spend an advertiser's budget optimally \emph{across a set of different products}, where each product may have a different user base, competition, and advertising dynamics. Large tech companies have platforms that allow the advertisers to set-up campaigns that can be delivered to users across different internal products (in this paper we will use the unified terminology of \emph{platforms}). The key challenge this leads to is that a campaign with a single budget has to acquire impressions from a variety of platforms that provide different value possibly due to a different user-base and different prices owing to different competition. Further challenges arise from the reality that advertisers often do not have a good understanding of their valuation of showing an ad to users spread across different platforms. Consider, for example, the case of an eCommerce business that advertises a specific good on their website. While they may have a good understanding of how much a conversion---i.e. the user purchases the good---is worth, they may have to pay for a click without knowing how likely the user is going to convert given a click on each of these platforms. Therefore, there is a need for developing budget management solutions that can optimize campaign delivery across multiple platforms in face of this uncertainty and budget constraints. 

	In this paper, we  provide a rigorous mathematical model to study the aforementioned application and provide algorithms that are optimal. We also expect these algorithms to be extremely fast and thus, can be deployed in real-world bidding engines where the system has only a few milliseconds to set bids. First, we describe some practical constraints that the bidding engines face before describing the various attempts at modeling this problem. In modern bidding engines, a bid is composed of multiple components. Most components of this bid is set for a \emph{batch} of requests, before adding a platform specific multiplier that depends on the request. Thus, from the perspective of the central bidding engine, the goal is to decide bids for each of the platforms and this bid will be applied to the next batch of requests in the auction stage. After setting this main bid, the engine only sees a feedback of the total value received from each platform and the total budget that is consumed.
	
	Motivated by the aforementioned example, in this work, we consider the \emph{multi-platform placement optimization with unknown valuations.} More specifically, we study the problem of an advertiser that wants to optimally spend a given budget $B$ over a time horizon $T$ across $m$ platforms. Those platforms select the winners via second-price auctions. The goal is to design a policy that submits, at each time $t$, $m$ bids $b_t(1), ..., b_t(m)$, one for each of the $m$-platforms. The bids are chosen such that the advertiser maximizes the cumulative reward over all the platforms while ensuring that the aggregate payments across all $m$ platforms and the time horizon $T$ do not exceed the allocated budget $B.$ In other words, the objective of the advertiser is to learn the valuation of bidding on each of those platforms while simultaneously maximizing their cumulative rewards. 
	
	The academic community has worked actively to provide methods for doing this effectively when a campaign delivers only in a single platform. The studied methods focused on solving the problem either by limiting the set of auctions that bidders are eligible for, e.g. \cite{mehta2007adwords, abrams2008comprehensive, azar2009gsp, goel2010gsp, karande2013pacing}, or by adjusting the bids that enter the auction, e.g. \cite{rusmevichientong2006, feldman2007budget, hosanagar2008, balseiro2017budget, conitzer2018pacing, ConitzerKPSSSW2019}. The latter are of particular interest, because they can either be implemented by the online advertising platform (as is common), but also by an advertiser themselves. However, the work on budget-management for online advertising has so far not focused on two important concerns: firstly, the marketing budget should be used over different platforms where users on different platforms are valued differently and platforms have different levels of competition, and secondly, advertisers may not have a clear understanding of their valuation for showing an ad to different users across these platforms.
	
	The problem of learning to bid without knowing your value naturally leads to a \emph{multi-armed bandit} formulation of the problem, and for a single platform has been studied without a budget constraint by Feng et al. \cite{feng2018learning}, and with a budget constraint by \cite{amin2012budget, tran2012knapsack, flajolet2017real}, though none of these approaches consider bidding on multiple platforms. In this work, we formulate the problem as a stochastic bandits with knapsacks problem (BwK), introduced in the seminal work by Badanidiyuru et al. \cite{BKS13} and generalize the approach to the multi-platform optimization problem (\cite{Karthik-aistats18}). One of the main contributions of the present paper is in highlighting the salient features of BwK that need to be adapted to design a computationally tractable algorithm that can compute the optimal bid vector from a combinatorial feasible bid set and ensuring that the performance (as measured by regret) of our algorithm scales optimally in both the number of platforms and the optimal value in hindsight. Sankararaman and Slivkins \cite{Karthik-aistats18} consider a more general problem where the combinatorial feasible set can be any general matroid. Thus, their algorithm is much more complicated and their bounds in the regime when $OPT \ll T$ is weaker than the one in the present paper. In this work, the feasible bid  set is the special case of partition matroid which our algorithm specifically exploits and thus, does not require the rounding procedure required in \cite{Karthik-aistats18}.

	\paragraph{Modeling Choices.}
	Before we present our modeling of this problem with an exponential set of bidding $m$-tuples, we like to discuss some alternative formulations and their drawbacks. A first natural attempt at modeling this problem is to set it up as a \emph{multi-armed bandit} problem where we have $n \times m$ arms with a group of $n$ arms denoting the set of discrete bids that can be chosen for each of the $m$ platforms. The game proceeds in $T \times m$ rounds, where $T$ corresponds to the total number of steps for which the bidding engine has to decide the main bid. We split the $ T \times m$ rounds into $T$ phases where within each phase the algorithm plays for $m$ rounds and sets the bid for each of the $m$ platforms sequentially. At the end of this phase, the algorithm receives a feedback, \emph{i.e.,} the total value and the budget consumed on each platform during this phase. The goal of the algorithm is to maximize the total value while respecting the total budget constraint. We argue that this model has two major drawbacks and thus, is not suitable for this application. First, at any given time-step $(t, i) \in ([T], [m])$, the algorithm cannot choose all the arms, but only a subset of the arms (\emph{i.e.,} those that correspond to platform $i$). In a typical multi-armed bandit problem, we assume that all arms are available at all time-steps. We could model it as the more general \emph{sleeping bandits} \cite{kleinberg2010regret} problem but note that even without budget constraints this problem requires a change in benchmark (\emph{i.e.,} the best ordering of arms) which does not translate to a meaningful notion in our application. Second, in $\MAB$ problems, we assume that the feedback received by the algorithm is immediate. However, in this modeling approach we would only receive the feedback at the end of the phase. Algorithms for \emph{delayed feedback} models compare against \emph{static} policies, while with budget constraints it is known that (\emph{e.g.,} \cite{BKS13}) optimal dynamic policies can obtain at least twice more reward than  the optimal static policy. 
	
	The second natural attempt to model this problem is to consider the \emph{contextual} bandit framework. We have $n$ arms (corresponding to each of the discrete bids) and $T$ time-steps. At each time-step, the algorithm first sees a context $x_t$ (in this case it corresponds to the platform on which the bids need to be chosen on) and then chose an arm that is a function of history and context. After every $m$ rounds, the algorithm receives the feedback for the last $m$ rounds (note even the bidding engine sets the bid for a platform, the feedback is collected over the batch and given back to the engine at the end of the phase). This modeling also has a two major drawbacks. First, as before, we need to handle delayed rewards (now in the much harder contextual bandit setup) for which to obtain algorithms with provable guarantees we need to make further assumptions (that may not hold in practice) such as linear rewards, static optimal policy. Second, contextual bandit problems with budget constraints are significantly harder (\emph{e.g.,} \cite{agrawal2016efficient}) where the regret bounds require the \emph{large} budget assumption. In the context of online advertising for any campaign the total budget is very small compared to the number of auctions it participates in and thus, such large budget assumption is not reasonable.
	
	Owing to the above difficulties, we model this problem as a variant of the \emph{Bandits with Knapsacks} problem with combinatorial actions. We assume that we have $n \times m$ arms, where a group of $n$ arms correspond to the possible bids for each of the $m$ platforms. At each time-step, the algorithm chooses one bid per platform (a total of $m$ arms) and sets this to be the bid for the next phase of auctions. At the end of this phase, the algorithm receives the feedback (the value and the budget consumed) and the algorithm can adapt its bid for the next phase. If for any platform, the bidding engine chooses not to bid for this phase, the algorithm simply chooses the $0$-bid which ensures that it will not win any auction. Moreover, in the generalization later we show that the bids need not be discrete and can be a continuous value in the range $[0, r]$ for some known value $r \geq 1$. This modeling removes the difficulties in the above natural approaches: the benchmark for an optimal dynamic policy is indeed a best possible algorithm the bidding engine can employ, the feedback for each of the arms is obtained immediately at the end of the phase and we do not need to make assumptions on the ratio of budget to the total number of time-steps. \emph{A priori}, a new challenge this modeling approach introduces is that the algorithm needs to choose from an exponentially available set of choices at each time-step. However, as \cite{Karthik-aistats18} show, both the algorithm and the regret bounds can overcome this and only depend on a polynomial function of the total number of platforms.

\subsection{Overview of  Contributions}\label{sec:contributions}

We present algorithms for budget management in online advertising across multiple bidding platforms. The paper presents five main contributions:

\begin{itemize}
	\item \textbf{Tractable mathematical model and algorithms.} The first contribution of this paper is to identify the correct mathematical model to study this problem. As mentioned earlier a number of natural choices fail since they lead to difficulties in obtaining an algorithm with provable guarantees. We study this problem in the framework of Stochastic bandits with knapsacks (BwK) introduced by Badanidiyuru et al. \cite{BKS13}. The goal is to design online bidding strategies that use the information available from the past rounds of ad auction to achieve suboptimal regret with respect to the optimal stochastic strategy that knows in advance the distributions of price and valuation on each platform. Unlike \cite{Karthik-aistats18} we do not require the algorithm to invoke the rounding sub-routine at each time-step. The rounding sub-routine runs in $\mathcal{O}((mn)^2)$ time-complexity per time-step and cannot be amortized across rounds. Thus, this improves the running time of our algorithm when the bid-space and the number of platforms become very large. On the other hand, similar to \cite{Karthik-aistats18}, our algorithm requires us to solve a linear program at each time-step. In practice, this can be fast, since we can cold-start the solution for each time-step with the solution of the previous time-step and thus, the computations can be amortized across the $T$ time-steps to an average of a constant number of steps per-round.
	
	\item {\bf Algorithm and regret analysis}. In Section~\ref{sec:algorithm}, we present a bandit algorithm that achieves a regret bound  $O\left(OPT \sqrt {\frac{mn}{B} }+ \sqrt{mn OPT}\right)$ for discrete bid spaces, (where $m$ is the number of platforms, $n$ the size of the discrete bid space, $B$ the budget and $OPT$ is the performance of the optimal algorithm that knows the distributions). This improves over the regret bound presented in \cite{Karthik-aistats18} from $\mathcal{O}(\sqrt{T})$ to $\mathcal{O}(\sqrt{OPT})$ while keeping the dependence on the other parameters same. Our approach follows very closely the Primal Dual algorithm with knapsack constraints of \cite{BKS13}.  However, we also show how to reduce the problem of finding the optimal bid vector to the problem of maximizing the ratio between a linear function of the UCB of the valuations and a linear function of the LCB of the costs, and this  problem can be solved in polynomial time for a set of totally unimodular linear constraints \cite{AVADHANULA2016612}. In Section~\ref{sec:continuous}, we use the discretization idea introduced in \cite{Badanidiyuru:2018}, to discretize the continuous bid space in $[p_0,1]$, with $p_0$  being the minimum reserve price across all platforms, by using an $\epsilon$ grid.  The regret bound $\tilde{O}\left(\frac{(mv_0)^{1/3}}{p_0^{2/3}} \cdot \min\left\{ \left( \frac{Bv_0}{p_0} \right)^{2/3}, (m T)^{2/3} \right\} \right)$ is obtained by adding the discretization error $\frac{B\epsilon v_0}{p_0^2}$, with $v_0$ being the maximum expected valuation across all platforms, to the regret  for the discrete setting with $n = 1/\epsilon$ bid values.  For the special case of single platform we improve over the $\mathcal{O}(T^{2/3})$ bounds in \cite{tran-tranh14} for the more challenging and realistic setting of $B= o(T)$. Additionally, the algorithms in \cite{tran-tranh14} require the information on whether the algorithm won the auction or not at each time-step; while our algorithms do not need this information. Note, that ours is a more realistic assumption, in general, since advertising systems have \emph{checks and balances} in place, after an auction phase and thus, impression/conversion information does not directly translate to whether an ad won an auction or not. 
	\item {\bf Lower bounds}. In Section~\ref{sec:lowerbound} we complement our algorithmic results with a $ \Omega\left (\sqrt {m OPT} \right)$  lower bound for the discrete case and an  $\Omega\left( m^{1/3} B^{2/3} \right)$  lower bound for the continuous case thus showing that our algorithmic results are close to optimum despite the complexity of our setting. 
	\item {\bf Experiments on real data}. In Section~\ref{sec:experiments} we evaluate our online bidding strategies on  real-world dataset with multiple platforms and several budget limited advertisers  obtained from the logs of a large internet advertising company.   We compare our algorithms against  the total reward of the optimal $LP$ value computed on the mean valuation and prices and on two different baselines. The first baseline is the naive \emph{UCB} algorithm that ignores the budget constraints and maximizes the rewards. The second algorithm is the \emph{LuekerLearn algorithm} from \cite{tran-tranh14} adapted to multiple platforms and stochastic valuations.  Our experimental results show that we outperform both the benchmarks in efficiency (total regret accumulated). We show how the various algorithms compare as a function of budget and number of platforms. Moreover, we also observe that our proposed algorithm uniformly depletes the budget and always runs till the specified time-limit, while both the baselines exhaust budget very early. Uniform budget spend is an \emph{expectation} many advertisers have and thus, this behavior of our algorithm is extremely desirable for practical deployment.
\end{itemize}

\subsection{Additional Related Work}
\label{sec:relatedwork}

The problem of optimal bidding policies in second price auctions from the perspective of bidders/advertisers can be broadly categorized into four themes.

\paragraph{Budget management with known valuation.} Early approaches to budget management with known valuations (and a single platform) have focused on shading an advertiser's value by a multiplier (often called the ``pacing multiplier''). For example, the work of \cite{rusmevichientong2006, feldman2007budget, hosanagar2008} implement bid shading strategies to maximize the advertisers ROI. What makes this approach appealing, is that the static problem (with known competitor bids) has an optimal solution with a \emph{constant} pacing multipier, so the dynamic approach aims to find this constant. More recent work \cite{cray2007,borgs2007dynamics,balseiro2017budget, conitzer2018pacing, ConitzerKPSSSW2019, balseiro2021contextual} have complemented this bid-manipulation literature with equilibrium analysis under the assumption that all advertisers use the same bid-shading approach. 

An alternative to bid modification, is for a platform to restrict the number of auctions that an advertiser participates in. Work by Mehta et al.~\cite{mehta2007adwords} give revenue guarantees for the online matching problem of users showing up over time. Subsequently bidder selection has been applied to multi-objective optimization for Generalized Second Price auctions in e.g. \cite{abrams2008comprehensive,azar2009gsp,goel2010gsp,karande2013pacing}. While the work on bidder selection yields allocations that respect the budget constraint, advertisers are not best-responding to the competition and they can generally improve their allocation by lowering their bids.

\paragraph{Bidding with unknown valuations and no budget constraints.} Feng et al. \cite{feng2018learning} focus on the problem of bidding on a platform without knowing the valuations for users. They model this as a stochastic bandit setting with partial feedback (they only learn the value if the ad won, and the action that the advertiser cares about was taken). They give an online learning algorithm with regret that grows logarithmically in the size of the action space. Their model considers more complex auction environments such as position auctions\footnote{Though not with advertiser-specific discount curves \cite{colini2020ad}\cite{elzayn2021equilibria}.} \cite{V07,EOS07}, though it doesn't take into account multiple platforms or budget constraints. Follow-up work by Feng et al.~\cite{feng2019online} focus on measuring incentive compatibility, but their learning setup also uncovers the best response for an unknown auction format as side effect.

\paragraph{Budget-constrained bidding and unknown valuations.} This line of literature combines the problem settings of the above two themes. More specifically, existing work consider the setting where the valuation as well as the distribution of the competing bids is unknown and focus on optimal bidding strategies that satisfy the specified budget, maximizes the cumulative rewards and learns the valuation of the item. To the best of our knowledge, the problem of learning the optimal bid in second price auctions was first considered by Amin et al. \cite{amin2012budget}. They focus on the discrete bids setting, where the feasible bids belong to a discrete set of bids and formulate the problem as Markov decision process with censored observations. They propose a learning algorithm based on the popular product-limit estimator and demonstrate empirically that their algorithm converges to optimal solution. Tran-Tranh et al. \cite{tran2012knapsack} followed up on the work of \cite{amin2012budget} and provided theoretical guarantees for the algorithm proposed in the later work. More specifically, they assume that the bids are sampled from a discrete distribution in $\{1, \cdots, C\}$ and establish $O(\sqrt{CB})$ regret bounds, where $C$ is the number of feasible bids and $B$ is the total budget. The aforementioned regret bounds depend on some parameters of the competing bids distribution. Flajolet and Jailiet \cite{flajolet2017real} consider a more general variant of the aforementioned problem by relaxing the assumption that the bids come from a discrete set  and proposed an algorithm based on the probabilistic bisection (see \cite{stochastic_root_finder}) and upper confidence bounds (UCB) family of algorithms. They establish $O(\sqrt{T})$ regret bound for the problem under the assumption that budget scales linearly with time, i.e. $B = \Theta(T)$. All three works (\cite{amin2012budget}, \cite{tran2012knapsack} and \cite{flajolet2017real}) in this theme only focus on bidding strategies for a single platform and cannot be easily generalized to the multi-platform setting.  Finally,  Nuara et al. \cite{aless2020online} study the problem of joint bid/daily budget optimization of pay-per-click advertising campaigns over multiple channels. The authors formulate the problem with combinatorial semi-bandits, which requires solving a special case of the Multiple-Choice Knapsack problem every day, and it obtains an $O(\sqrt T)$ regret bound.  Differently from our work that performs on online budget optimization across multiple rounds, this last paper does static budget optimization in each single round. 

\paragraph{Stochastic bandits with knapsack constraints. } Stochastic bandits with knapsacks (BwK), introduced in the seminal work of Badanidiyuru et al. \cite{BKS13} is a general framework that considers the standard multi-armed bandit problem under the presence of additional budget/resource constraints in addition to the time horizon constraint. The BwK problem encapsulates a large number of constrained bandit problems that naturally arise in many application domains including dynamic pricing, auction bidding, routing and scheduling. \cite{BKS13, Badanidiyuru:2018} presents two algorithms based respectively on balanced exploration and primal-dual  also used in this work with performances close to the information theoretical optimum. More specifically, they establish a $O(\log(dT)(\sqrt{X \bar{{\sf OPT}}}  + \bar {\sf OPT}\sqrt{m/B}))$, where $X$ is the total number of arms, $d$ is the number of constraints, $B$ is the binding budget constraint and $\bar{\sf{OPT}} $ is a ``reasonable" upper bound on the cumulative optimal reward. A naive generalization of the BwK framework to our problem would result in a combinatorial number of arms ($n^m$ possible bids). Recently, Sankarraman and Slivkins \cite{Karthik-aistats18} extended the BwK framework to combinatorial action sets and established $O(\sqrt{XT})$ regret bounds, which is only optimal only when the cumulative reward is comparable to the length of the time horizon, which is not the case in many online auctions. In this work, we extend the BwK (\cite{BKS13}) framework to propose an optimal bidding strategy with performances guarantees that are near optimal even when the cumulative reward is much smaller in comparison to length of the time horizon. 

The BwK framework \cite{BKS13} also has  been extended with near-optimal regret bounds to  contextual bandits (see \cite{pmlr-v35-badanidiyuru14,CBwK-colt16,CBwK-nips16}), settings with concave reward functions and convex constraints (see \cite{AD14}). Adversarial bandits with knapsack are also  extended in \cite{Immorlica19} to the combinatorial semi-bandit, contextual, and convex optimization settings.  However, these problem settings are not immediately applicable to our problem.

\section{Ad Platform Optimization with Unknown Valuations}\label{sec:problemformulation}

We model multi-platform real-time bidding for budgeted advertisers as a multi-armed bandit problem with knapsack constraints  \cite{Badanidiyuru:2018}. The advertiser has an available budget $B$ for $T$ rounds of the ad auction.  At each time step $t\in T$, the bidder chooses a bid $\vec{b}_t \in [0, 1]^m$ where $b_t(i)$ is the bid for platform $i  \in [m]$. We also consider the special case of discrete bid in section \ref{sec:algorithm}, each of the bid $b_t(j)$ is chosen from the finite set $\mathcal{B} = \{b_1,\ldots, b_n\}$, where $b_j\in [0,1], j\in [n]$. The set  of arms  $\mathcal X = \{ \langle b(1), \ldots, b(m)\rangle, b(i)\in \mathcal{B}, i\in [m]\}$ is therefore given by the set of distinct bidding $m$-tuples. Denote by  $\bm b= \langle b(1), \ldots, b(m)\rangle$ the generic bid vector and  by $x_{\bm b} \in \mathcal{X}$ the corresponding arm. 

At each phase $t$, each platform $i$ runs many second price auctions among the bids that are received for this phase. When at time $t$, the bidder bids ${\bf b}_t$ across the platforms, bid $b_t(i)$ is entered on platform $i$. Let $p_t(i)\sim P(i)$ be the \emph{critical bid}\footnote{The critical bid is either the highest bid of the other bidders or the reserve price, whichever is higher.} for the bidder on platform $i$ at time $t$ (where $P(i)$ is a stationary distribution). If $b_t(i) \geq p_t(i)$ the bidder wins the auction (we break ties in favor of the advertiser), the price is equal to the critical bid $p_t(i)$ and at that point learns the realized value $v_t(i) \sim V(i)$ (where $V(i)$ is a stationary distribution). For ease of presentation, bids, critical bids and values are normalized so that $b(i)$, $p(i)$, $v(i) \in  [0, 1]$ for all $i\in[m]$, and we'll use the terms critical bid and price interchangeably where appropriate \footnote{Note this is without loss of generality, since we can scale the values and multiply this scale in the regret bound.}.

We denote by 
\[
r_{x(\bm b)}= \sum_{i\in [m]} v(i) \mathbb{I} [b(i)\geq p(i)]\qquad\text{and} 
\]
\[
c_{x(\bm b)} = \sum_{i\in [m]} p(i) \mathbb{I} [b(i)\geq p(i)]
\]
the reward and the  cost of arm $x_{\bm b}$ with critical bid vector $\bm p$ and reward vector $\bm v$.  
We denote  by  $\bm b_t= \langle b_t(1), \ldots, b_t(m)\rangle$ the bid vector of the bidder at time $t$.  
We also use  $\bm p_t = \langle p_t(1), \ldots, p_t(m)\rangle$ for the critical bid vector that is realized at time $t$ and $\bm v_t = \langle v_t(1), \ldots, v_t(m)\rangle$ for the  vector of rewards that is realized at time $t$.  

In the stochastic bandit setting, critical bid $p(i)$ and reward values $v(i)$ are drawn at any time $t$, respectively,  from the stationary  independent
distributions  $P(i), V(i)$ for $ i\in [m]$, unknown to the algorithm. 
We also denote with $\bar r_{x(\bm b)} = E[r_{x(\bm b)}]$ and $\bar c_{x(\bm b)}= E[c_{x(\bm b)}]$ the expected value of 
the reward and the cost of arm $x(\bm b)$, where the expectation is taken over the critical bid $p(i)\sim P(i)$ and the reward value $v(i)\sim V(i)$. 

The bidder must decide on the bidding vector $\bm b_t$ without the knowledge of the critical bid vector $\bm p_t$ and the reward vector $\bm v_t$. 
If $b_t(i)\geq p_t(i)$, the feedback provided to the advertiser upon bidding $b_t(i)$ is price $p_t(i)$ and utility $v_t(i)$.  If $b_t(i)< p_t(i)$, the advertiser only learned that $b_t(i)$ is lower than the critical bid. 
The goal is to design an online bidding strategy that selects  a bid vector $\bm b_t$ at each round $t$ 
so that the time-averaged reward across the multiple platforms is maximized and the total cost paid by the algorithm does not exceed budget $B$.

Given that the advertiser is budget limited,  we use random variable  $\tau$   to indicate the first time the budget of the advertiser is exhausted, i.e., 
\[
\tau = {\rm min} ~ \left( T+1, \min \left\{z \in [T] | \sum_{t=1}^z c_{x(\bm b_t)} >B \right\} \right)
\]
We compare the reward obtained by the algorithm against the reward $OPT$ of an optimal bidding strategy that knows the  distributions 
$P(i), V(i)$ for $i\in [m]$,  and it is allowed to submit a new bid vector at each time.
Our goal is to bound the regret of the bidding strategy of the advertiser defined by 
\[
OPT - \mathbb{E}\left[ \sum_{t=1}^\tau r_{x(\bm b_t)} \right],
\]
where the expectation is taken over the random variables $\tau, p(i), v(i), i\in [m]$.

\section{Discrete Bid Spaces}
\label{sec:algorithm}

We adapt the Primal Dual algorithm for bandits with knapsack constraints of \cite{Badanidiyuru:2018}. We use two resources in our formulation.
The budget $B$  is consumed at rate equal to the price paid for each ad auction that is won by the advertiser.  Furthermore, we have a second resource with budget $B$ that consumes deterministically a budget $B/T$ at each round of ad auction. The second resource ensures that the system can bid at most $T$ times on each platform.

For the regret analysis of the algorithm, following the primal dual approach of  \cite{Badanidiyuru:2018},  we compare with an upper bound on OPT  given by the following LP formulation: 

\begin{eqnarray}
OPT_{LP} =    {\rm max} ~ \sum_{x\in \mathcal{X}} \xi_x ~ \bar r_x  \\
 \sum_{x\in \mathcal{X} } \xi_x ~ \bar c_x &\leq&   B\\
  \sum_{x\in \mathcal{X} } \frac{B}{T} \xi_x &\leq&  B\\
\xi_x &\geq& 0, \forall x \in X 
\end{eqnarray}

with $\xi_x$ being the number of times arm $x\in \mathcal{X}$ is used during the $T$ rounds.

According to  the following claim, we can use the optimal LP solution in place of $OPT$ for bounding the regret of the algorithm:
\begin{claim}
$OPT_{LP}$ is an upper bound on the value of the optimal dynamic policy: $OPT_{LP} \geq OPT$.
\end{claim}

In the algorithm, we denote by $\bar v_{ij}$ and $\bar c_{ij}$ the expected reward and the expected cost, respectively, 
obtained by bidding value $b_j$ on platform $i$.  We also denote by $v_{ij}^{UCB}(t)$ the upper confidence bound estimation of $\bar v_{ij}$ 
and by  $c^{LCB}_{ij}$ the lower confidence bound estimation of $\bar c_{ij}$.  Concretely, they are defined as follows. Here $N_{ij}(t)$ denotes the number of times we bid $i$ on platform $j$ until time $t$ and $C_{rad} = \Theta(\log mnT)$.
	\begin{align*}
		v_{ij}^{UCB}(t) & := \frac{1}{N_{ij}(t)} \sum_{t' < t} v_{ij}(t) \cdot \mathbb{I}[\text{bid $i$ on platform $j$}]\\
		 & \qquad \qquad + \sqrt{\frac{C_{rad}\cdot \bar v_{ij}}{N_{ij}(t)}} + \frac{C_{rad}}{N_{ij}(t)} \\
		c_{ij}^{LCB}(t) & := \frac{1}{N_{ij}(t)} \sum_{t' < t} c_{ij}(t) \cdot \mathbb{I}[\text{bid $i$ on platform $j$}] \\
		& \qquad \qquad - \sqrt{\frac{C_{rad}\cdot \bar c_{ij}}{N_{ij}(t)}} - \frac{C_{rad}}{N_{ij}(t)}
	\end{align*}

We  denote by $\lambda_t(1)$ and $\lambda_t(2)$ the estimation 
of the dual variables computed by the algorithm. We have $d=2$ in our case but we leave a generic term $d$ for the purpose of an easy generalization to the case of an individual budget $B_i$ for each $i\in [m]$.

\begin{algorithm}
\caption{Multi-platform bidding with parameter $\epsilon \in (0, 1)$}
\begin{algorithmic}

\STATE{\bf INITIALIZATION}
	\STATE{Play arm $x_{\bm b_j}$ with bid $\bm b_j =  \{b(i)=b_j \}_{i\in [m]}, j\in n$ for an initial UCB estimate of $\bar r_{ij}$ and an initial LCB estimate of $\bar c_{ij}$ }
	\STATE{$\lambda_1(1)=1$, $\lambda_1(2)=1$  is the estimate for $\eta_1$ and $\eta_2$}
	\STATE{set $\epsilon = \sqrt{\ln (d) /B}$}
	\FOR{$t=n+1, \ldots, \tau $ (i.e., until resource budget is exhausted)}
		\STATE{Obtain the maximizer to the following
		\begin{eqnarray*} 
		& \alpha_{i,j_i^*} = \\
		 & \argmax_{\alpha_{ij}\in \{0,1\}} & \frac{\sum_{i\in [m]} \sum_{j\in [n]} r_{ij}^{UCB} \alpha_{ij}}
        		{ \lambda_t(1) \sum_{i\in [m]} \sum_{j\in [n]} c_{ij}^{LCB} \alpha_{ij} + \lambda_t(2)\frac{B}{T}} \\
		& {\rm s.t.} \\
		& \sum_{j\in n} \alpha_{ij}=1, \forall i\in [m]
		\end{eqnarray*} }
		\STATE{Play arm $x_{\bm b_t}$ with bid $\bm b_t =  \{b_{i,j_i^*} \}_{i\in [m]}$}
		\STATE{Update UCB estimate $r_{ij}^{UCB}$ and LCB estimate $c_{ij}^{LCB}$}
		\STATE{Compute LCB estimate cost of arm $x_t$: $c_{1x_t}^{LCB} = \sum_{i\in m} c_{i, j_i^*}^{LCB}$, $c_{2x_t}=B/T$}
		\STATE{Update estimate of dual variables: }
		\STATE{$\lambda_{t+1}(1) = \lambda_{t}(1) (1 + \epsilon)^{c_{1x_t}^{LCB}}$}
		\STATE{$\lambda_{t+1}(2) = \lambda_{t} (1)(1 + \epsilon)^{c_{2x_t}}$}
	\ENDFOR	
	
\end{algorithmic}

\label{alg:MultiBidding}

\end{algorithm}

In the initial exploration phase, bid $b_i$ is submitted on all platforms at time $t=i$. The cost of the initial phase is therefore bounded by $mn$. 
At each round we  compute the arm with bid vector that maximizes the \emph{bang-per-buck} ratio between the Upper Confidence Bound of the reward and the Lower Confidence Bound of the  normalized cost. The number of different arms that is exponential can be reduced in the analysis by pruning out suboptimal arms. The optimal arm according to the $UCB$ and $LCB$ approximations can actually be computed in polynomial time since this is the problem of 
optimizing a rational function subject to a set of linear constraints described by a totally unimodular matrix \cite{AVADHANULA2016612}. After the feedback is received, the UCB estimation of the rewards and the LCB estimation of the costs are updated. Variables $\lambda(1), \lambda(2)$ are estimated using multiplicative weight update \cite{LW94}.

\subsection{Analysis of the algorithm.}
\label{regret}
\label{sec:analysis}

Once the problem of selecting the arm has been addressed through a separate optimization step, the analysis of the algorithm follows very closely the one of  \cite{Badanidiyuru:2018} with some care that allows to replace the exponential number of arms with $mn$ in the regret bound. 

Let $y_t(i) = \lambda_t(i)/{\parallel \lambda_t \parallel}_1, i=1,\ldots, d$ be the normalized cost of the resources.  For a parameter $\epsilon \in [0,1]$,  for every vector $y$, for any sequence of payoff vectors $c_1, \ldots, c_{\tau} \in [0,1]^d$, Hedge's guarantee gives \cite{LW94}: 

\begin{eqnarray}
\label{eq:Hedge}
\sum_{t=1}^\tau y_t^T c_t \geq (1-\epsilon) \sum_{t=1}^\tau y^T c_t  - \frac{\ln d}{\epsilon}.
\end{eqnarray}

In what follows, we  denote by $ c_{x_t} = \left[
                \begin{array}{ll}
                  c_{1x_t} \\
                  c_{2x_t}\\
                 \end{array}
              \right],$ the cost vector of arm $x_t$.

We must  consider in the analysis the error that derives from using the UCB estimate for the rewards and LCB estimate for the costs.  
First of all, given that the maximum reward is $mT$, if we fail to have a clean execution, the loss is $O(mT)$.  
If we select $C_{rad} = \Theta(\log d T m)$, the  probability of failure can be made much smaller than $1/(mT)$.  


Let us denote by $l_{x_t}$ the LCB estimate of $c_{x_t}$ and by $u_{x_t}$ the UCB estimate of $r_{x_t}$. 
Moreover, let $E_t=c_{x_t}-l_{x_t}$ be the error on the cost and $\delta_t = u_{x_t} - r_{x_t}$ the error on the reward. 
We also denote by $REW^{UCB} = \sum_t u^T_{x_t} \xi_t$ the UCB reward of the algorithm. 
 
\begin{claim}
\label{claim:lbrewucb}
\begin{eqnarray*}
	& REW \geq OPT_{LP} \left[ (1-\epsilon) - \frac{mn +1}{B} - \frac{1}{B} ||\sum_{1<t<\tau}  E^T_t \xi_t ||_{\infty}  - \frac{1}{B} \frac{\ln d}{\epsilon } \right] \\
	& - \left| \sum_{1<t<\tau} \delta_t^T \xi_t \right|.
\end{eqnarray*}
\end{claim}
\begin{proof}
The structure of this proof is similar to \cite{Badanidiyuru:2018}.
Let  $\bar y = \frac{1}{REW^{UCB}} \sum_{n<t<\tau} (u^T_{x_t} \xi_t ) y_t$
We prove the following inequalities:

\begin{eqnarray*}
B & \geq & \textstyle {\bar y}^T c \xi^* \\ 
    & = &  \textstyle  \frac{1}{REW^{UCB}} \sum_{n<t<\tau} (u^T_{x_t} \xi_t)(y_t^T c \xi^*) \\ 
   &\geq & \textstyle \frac{1}{REW^{UCB}} \sum_{n<t<\tau} (u^T_{x_t} \xi_t)(y_t^T l_{x_t} \xi^*)  \\
                 &\geq &\textstyle \frac{1}{REW^{UCB}} \sum_{n<t<\tau} (u^T_{x_t} \xi^*)(y_t^T l_{x_t} \xi_t)\\  
   &\geq &\textstyle   \frac{1}{REW^{UCB}} \sum_{n<t<\tau} (r^T  \xi^*)(y_t^T l_{x_t} \xi_t) \\
 &\geq &\textstyle  \frac{OPT_{LP}}{REW^{UCB}} \left[(1-\epsilon) y^T \left( \sum_{n<t<\tau}  l_{x_t} \xi_t  \right) -\frac{\ln d}{\epsilon} \right] \\
 &\geq &\textstyle  (1-\epsilon) \frac{OPT_{LP}}{REW^{UCB}}  \\
 && \textstyle \left[ y^T \left( \sum_{n<t<\tau}  c_{x_t} \xi_t  \right) - y^T\left( \sum_{n<t<\tau}  E^T_t \xi_t  \right) -  \frac{\ln d}{\epsilon} \right] \\
 &\geq & \textstyle  \frac{OPT_{LP}}{REW^{UCB}} \\ 
 && \textstyle \left[ (1-\epsilon) (B-mn-1) - (1-\epsilon)  y^T \left( \sum_{n<t<\tau}  E^T_t \xi_t  \right) -  \frac{\ln d}{\epsilon} \right],
\end{eqnarray*}

The first inequality follows from primal feasibility, the second inequality by the definition of $\bar y$, the third inequality by clean execution, the forth inequality by the rule of selection of the arm, the fifth inequality follows  from clean execution, and  the sixth inequality from the guarantee of Hedge of equation \ref{eq:Hedge}.

For bounding the regret of the algorithm we finally use: 

\[
REW \geq REW^{UCB} - \sum_{n<t<\tau} (u_{x_t} - r_{x_t})^T \xi_t = REW^{UCB} - \left| \sum_{1<t<\tau} \delta_t^T \xi_t \right|.
\]
\end{proof}


We combine this claim with the following two bounds which can again be derived using the approach in \cite{BKS13}.
\begin{equation*}
	\left|\sum_{n<t<\tau} \delta_t^T \xi_t \right| = O\left( \sqrt{C_{rad} mn REW} + C_{rad} mn \ln T\right), 
\end{equation*}
\begin{equation*}
	\left |\left|\sum_{1<t<\tau}  E_t \xi_t \right |\right|_{\infty} = O\left( \sqrt{C_{rad}  mn B} + C_{rad} m n \ln T \right).
\end{equation*}
	By assuming $mn<B$ and $\epsilon = \sqrt{\frac{\ln d}{B}}$, we conclude with the following theorem:.

\begin{theorem}
\label{thm:maindiscrete}
The regret of the algorithm is bounded by 
	\[O\left(OPT_{LP} \sqrt {\frac{mn}{B} }+ \sqrt{mn OPT}\right).\]
\end{theorem}

\begin{proof}
	We use for the proof the  following two claims on the UCB estimate of the reward and the LCB estimate of the cost  are proved in \cite{Badanidiyuru:2018}:

\begin{claim}
\label{claim:error}
$\left|\sum_{n<t<\tau} \delta_t^T \xi_t \right| = O\left( \sqrt{C_{rad} mn REW} + C_{rad} mn \ln T\right) $
\end{claim}

\begin{claim}
$||\sum_{1<t<\tau}  E_t \xi_t ||_{\infty} = O\left( \sqrt{C_{rad}  mn B} + C_{rad} m n \ln T \right) $.
\end{claim}

We start from the claim of Claim \ref{claim:lbrewucb}:

\begin{eqnarray*}
REW &\geq & \textstyle  OPT_{LP} \left[ (1-\epsilon) - \frac{mn +1}{B} - \frac{1}{B} ||\sum_{1<t<\tau}  E_t \xi_t ||_{\infty} -  \frac{\ln d}{\epsilon B} \right] \\
&& \textstyle - \left| \sum_{1<t<\tau} \delta_t^T \xi_t \right| 
\end{eqnarray*}

%

By assuming $mn< B/\ln d T$ and  by choosing $\epsilon = \sqrt{\frac{\ln d}{B}}$,  we bound the following three terms of the regret:

\begin{eqnarray*}
OPT_{LP} \left[ \frac{mn+1}{B} +  \frac{\ln d}{\epsilon B} \right] = O\left( OPT_{LP} \left[ \frac{mn}{B} + \sqrt{\frac{\ln d}{B}}\right]\right);
\end{eqnarray*}

\begin{eqnarray*}
&& \frac{OPT_{LP} }{B} ||\sum_{1<t<\tau}  E_t \xi_t ||_{\infty} \\
&=& O\left( \frac{OPT_{LP}}{B} \left[ \sqrt{C_{rad}  mn B} + C_{rad} mn \ln T \right] \right) \\
&=& O\left( \frac{OPT_{LP}}{\sqrt B} \sqrt{C_{rad}  mn} \right)
\end{eqnarray*}

and 
\begin{eqnarray*}
\textstyle \left| \sum_{1<t<\tau} \delta_t^T \xi_t\right| = O\left( \sqrt{C_{rad} mn REW} + C_{rad} mn (\ln T) \right)
\end{eqnarray*}

thus proving the claim. 
\end{proof}

\section{Continuous Bid Spaces}
\label{sec:continuous}

In the previous section we considered the discrete setting with $n$ different bid values available  on each platform for the advertiser.  In this section we  consider the continuous setting, with  prices and  valuations being real values in $[0, 1]$. Our approach will be to discretize the continuous bidding space to consider bid values that are multiple of some small value $\epsilon$. The discretization of the bidding space will result into an additional regret loss.  Most of the following analysis is therefore concerned with bounding the error of the discretization process. 

Let the support of the critical bid distribution be in the interval $[p_0,1]$, with $p_0$ being a small constant that can be considered as the reserve price for the ad auction. Let $r_{b}(i)$ and $c_{b}(i)$ the expected reward and the expected cost of bid $b$ on platform $i\in [m]$, where the expectation is taken over the critical bid $p(i)\sim P(i)$ and the reward values $v(i)\sim V(i)$. We also denote by $v_0$ the maximum expected reward over all platforms if the auction is won. 
 
The following Lemma shows that the discretization of the continuous bidding space is performed  at the expense of a limited additive loss in the buck-per-bang ratio. 

\begin{lemma}
\label{lemma:discretization}
For each bid $b\geq p_0$, and for each platform $i\in [m]$, it holds: 

\begin{enumerate}
\item  $c_{b+\epsilon}(i) \geq c_{b}(i)$ 
\item $\frac{r_{b+\epsilon}(i)}{c_{b+\epsilon}(i)} - \frac{r_{b}(i)}{c_{b}(i)} \geq \frac{\epsilon v_0 }{p_0^2}$
\end{enumerate}
\end{lemma}

\begin{proof}
The first part of the claim is straightforward since the expected cost can only increase with the value of the bid. 
We omit $i$ for the proof of the second part of the claim. 

Let $f(p)$ be the continuous density function of the cumulative distribution $P$. For the expected cost of bid $b$, it holds

\begin{eqnarray*}
c_b = \int_{p_o}^b p f(p) dp \geq  r_b p_0/v, 
\end{eqnarray*}
since the ratio between cost and reward of bid $b$ is at least the minimum cost $p_0$ divided by the  reward $v$ that 
is obtained if the bid is accepted.  

Denote $a=\int_{b}^{b+\epsilon} f(p) dp$. We obtain the following expressions for the  cost and the reward  at bid $b+\epsilon$:  
\begin{eqnarray*}
r_{b+\epsilon} &=& r_b + v \int_{b}^{b+\epsilon} f(p) dp = r_b + a \cdot v \\ 
c_{b+\epsilon} &=& c_b + \int_{b}^{b+\epsilon} p  f(p) dp \leq c_b + (b+\epsilon) a
\end{eqnarray*}

We therefore have: 

\begin{eqnarray*}
\frac{r_b}{c_b} - \frac{r_{b+\epsilon}}{c_{b+\epsilon}} & \leq& \frac{r_b}{c_b} -  \frac{r_b + a \cdot v}{c_{b} +  (b+\epsilon) a} \\ 
										& =& \frac{r_b  (b+\epsilon) a - a c_b v }{c_b (c_b + (b+\epsilon) a)}\\
										&\leq&  \frac{r_b  (b+\epsilon) a - a p_0 r_b  }{ r_b p_0 /v ( p_0  + (b+\epsilon) a)} \\
										&=& \frac{a ((b+\epsilon) -  p_0)}{p_0/v  (p_0 + (b+\epsilon) a)} \\
										&\leq& \frac{a v_0 }{p_0^2},
\end{eqnarray*}

with the previous inequalities that follow from $c_b \geq r_b p_0/v$, $c_b\geq p_0$, and $v\leq v_0$.  By continuity of $f(b)$,  $a = \int_{b}^{b+\epsilon} f(p) dp \leq \epsilon$ proves the result. 
\end{proof}

The total loss in the regret due to the discretization error of the buck-per-bang ratio is therefore upper-bounded by $\frac{B\epsilon v_0}{p_0^2}$ to be added to the 
regret bound of Theorem \ref{thm:maindiscrete}.  The set of distinct bid values  is $b_j = j \epsilon, j\in [n]$, and therefore $n$ can be replaced by $1/\epsilon$ in the claim of Theorem \ref{thm:maindiscrete}. 
	
We therefore conclude with the following theorem:

\begin{theorem}
	\label{thm:fullRegret}
The regret of the algorithm with discretized bids is bounded by 

\[ 
	\textstyle \tilde{O}\left(\frac{(mv_0)^{1/3}}{p_0^{2/3}} \cdot \min\left\{ \left( \frac{Bv_0}{p_0} \right)^{2/3}, (m T)^{2/3} \right\} \right)
\]
\end{theorem}
\begin{proof}
By substituting $n=1/\epsilon$ in the regret bound of the previous section, and by adding the additional loss in revenue given to the discretization, 
the bound on the regret is:

\[
	\tilde{O}\left(OPT_{LP} \sqrt {\frac{m}{B \epsilon} }+ \sqrt{\frac{m OPT}{\epsilon}} + \frac{B\epsilon v_0 }{p_0^2} \right).
\]

Note that $OPT\leq OPT_{LP} \leq \min\left\{ \frac{B v_0}{p_0}, mT \right\}$. We have two cases.

When $\min\left\{ \frac{B v_0}{p_0}, mT \right\} = \frac{B v_0}{p_0}$ the regret is bounded by:

\[
	\tilde{O}\left( \frac{v_0}{p_0}\sqrt {\frac{B m}{\epsilon} }+ \sqrt{\frac{m B v_0}{p_0 \epsilon}} + \frac{B\epsilon v_0}{p_0^2} \right).
\]

Substituting $\epsilon = \frac{p_0^{2/3} \cdot m^{1/3}}{B^{1/3}}$ we obtain that the regret is upper-bounded by $\tilde{O}\left(\frac{m^{1/3} B^{2/3} v_0 }{p_0^{4/3}}\right)$. 

Likewise, when $\min\left\{ \frac{B}{p_0}, mT \right\} = mT$, given that $mT<B$ if the problem is budget constrained, by setting  $\epsilon = \frac{m p_0^{4/3} T^{2/3}}{B v_0^{2/3}}$  we get the regret to be upper-bounded by $\tilde{O}\left(\frac{m T^{2/3} v_0^{1/3}}{p_0^{2/3}}\right)$ thus proving the theorem. 	
\end{proof}

%

\section{Lower Bounds}
\label{sec:lowerbound}

In this section, we show that the algorithms for discrete and continuous bid spaces are near-optimal. We start with a lower bound for discrete bid space as a function of $m$ and $OPT$. 
	\begin{theorem}
		\label{thm:LB2}
		For discrete bid-spaces, there exists an instance $\mathcal{I}$ such that any algorithm will incur a regret of at least $\Omega\left( \sqrt{m OPT} \right)$.
	\end{theorem}
	\begin{proof}[Proof Sketch.] The lower bound follows by adapting the classical lower bound for stochastic bandits \cite{Auer2002}. We consider one arm for each platform and a time horizon $T= 2 B.$
 Each platform $i\in [m]$ different from $j$  has expected reward $r_i=1/2$ and fixed cost $c_i=1/2$, while platform $j$ has expected reward $r_j = 1/2 (1+\epsilon)$ with fixed cost $c_j=1/2$.  
 Each platform  needs to be executed  $1/\epsilon^2$  times in order to find out the best arm.  The total budget needed in order to find the best arm is therefore equal to $B$. 
The  regret for all arms is $\Omega (m / \epsilon)$.  By setting $\epsilon = \sqrt{m/B}$ we obtain the lower bound since  the optimal stochastic policy will play arm $j$ for all the $2B$ rounds 
with cost $B$ and optimal reward   $OPT = (1+\epsilon) B$. 
	\end{proof}
\begin{theorem}
		\label{thm:LB1}
		For continuous bid spaces, there exists an instance $\mathcal{I}$ such that any algorithm will incur a regret of at least $\Omega\left( m^{1/3} B^{2/3} \right)$. 
	\end{theorem}
	\begin{proof}[Proof Sketch.]
		We start with the simple case of $m=1$. The proof of this theorem is derived by using the lower-bound construction for for Lipshitz bandits (see Chapter 4 in \cite{MAL-068}). In particular, we consider the simplest setting of $1$-platform with no restrictions on the budget constraint (\emph{i.e.,} $B = T$). The objective function of the algorithm thus, is to maximize the function $f(p)$, where $f$ is the continuous density function of the cumulative distribution of the critical bids $P$. Since, this is continuous this also implies that this function is $1$-Lipshitz and thus, the setting reduces to that of Lipshitz bandits. From Theorem 4.2 in \cite{MAL-068} we have that any algorithm incurs a regret of at-least $\Omega(T^{2/3})$. 
		
		Consider the case of $m$ platforms. The proof uses a similar strategy as in the lower bound for Lipshitz bandits (see Chapter 4 in \cite{MAL-068}) combined with the lower-bound strategy used for semi-bandits (\emph{e.g.,} section 6 in \cite{kveton2015tight}). We will closely follow the notations used in \cite{MAL-068}. Define the instance $\mathcal{I}(x^*, \epsilon)$ by the following. $\mu(x, i)$ denotes the mean reward for arm $x \in [0, 1]$ for platform $i \in [m]$. In a given instance all the platforms have the same mean reward function.
			\[
				\mu(x, i) = 				
				\begin{cases}
					\frac{1}{2} & \text{all arms x such that $|x - x^*| \geq \epsilon$}\\
					\frac{1}{2} + \epsilon - |x - x^*| & \text{otherwise}
				\end{cases}
			\]
			Similar to \cite{MAL-068}, we will now construct instances $\mathcal{J}(a^*, \epsilon)$ which is a semi-bandit problem on finite number of arms. More precisely, fix $K \in \mathbb{N}$ to be fixed in the analysis. We define a semi-bandit problem on $K*m$ atoms indexed as $(k_i, m_j)$ for $i \in [K]$ and $j \in [m]$. At each time-step the algorithm can choose exactly one of the atoms from the subset of atoms $\{(k_i, m_j)\}_{i \in [K]}$ for each $j \in [m]$. Thus, at each time-step the algorithm chooses atmost $m$ atoms. For any fixed $k_i$ for $i \in [K]$, the reward assigned for the atoms $\{(k_i, m_j)\}_{j \in [m]}$ is the same and set as in the lower-bound proof for Lipshitz bandits~\cite{MAL-068}. 
			
			We now use the observation made in \cite{kveton2015tight}; an instance with $m$ copies of a $K$-armed bandit problem (\emph{i.e.,} each arm is copied $m$ times) is equivalent to a single $K$-armed bandit problem where the reward is scaled by a factor $m$. Thus, the instance $\mathcal{J}(a^*, \epsilon)$ can be replaced by another instance $\mathcal{J'}(a^*, \epsilon)$ on $K$ arms, such that the reward is scaled by a factor $m$. This leads us to the following equivalent version of Theorem 4.2 in \cite{MAL-068}.
			
			\begin{theorem}[Theorem 4.2 from \cite{MAL-068}]
				For a stochastic multi-armed bandit problem with $K$ arms and time-horizon, with rewards in the range $[0, m]$. Let $ALG$ be any algorithm for this problem. Pick a positive constant $\epsilon \leq \sqrt{\frac{c m K}{T}}$ where $c$ is some absolute constant. Then there exists an instant $\mathcal{J'} = \mathcal{J'}(a^*, \epsilon)$ for $a^* \in [K]$ such that 
				\[
						\mathbb{E}[R(T)~|~\mathcal{J'}] \geq \Omega(\epsilon T).
				\]	
				
			\end{theorem}
			
			Thus, choosing $\epsilon = \frac{1}{2K}$ and $K = \left( \frac{T}{c m} \right)^{1/3}$ and proceeding as in \cite{MAL-068} we obtain a regret lower-bound of $\Omega(m^{1/3} T^{2/3})$ as claimed in the theorem.
	\end{proof}

\section{Experiments}
\label{sec:experiments}
 	\paragraph{Setup.}
		We construct a real-world dataset obtained from the logs of a large internet advertising company. In this dataset, we have $10$ platforms and $3$ budget constrained advertisers. For each advertiser, we normalize the dataset and obtain the relevant distributions for the critical bid $P$ and the valuation $V$. Our final dataset is obtained by sampling the critical bids from $P$ and the valuations from $V$ at each time-step. Due to the nature of the bidding system, we have a good approximation to the \emph{true valuations} for each advertiser on each platform. Outside of the larger advertisers, most advertisers have much smaller budgets compared to their total potential audience and thus, the most interesting regime is when $B \ll T$.
 
 	We run the various algorithms for $T=10^5$ time-steps with a pre-determined discretization of the bid space. Due to the nature of the price distributions for the various platforms (see Figure~\ref{fig:representativeDistributions} for representative distributions), we choose  the $\epsilon$-hyperbolic mesh as the discretization, where the bid is of the form $\frac{1}{1 + \epsilon \cdot \ell}$ where $\ell \in \mathbb{N}$. This is similar to the discretization used for dynamic procurement in \cite{Badanidiyuru:2018}. We vary the budget $B$ and compare the total reward obtained by the different algorithms against the optimal $LP$ value computed on the mean valuation and prices. For each setting of the parameters, we run $5$ independent runs and compute the average as the reward. Additionally, we also measure the run-time of our algorithm and compare that against the baseline LuekerLearn algorithm.
 	
 	\paragraph{Algorithms.}
 		We compare our algorithm against three different baselines. The first baseline is the SemiBwK-RRS algorithm proposed in \cite{Karthik-aistats18}. As mentioned in the introduction, on the theoretical front, our algorithm improves over this algorithm in worst-case scenario when $OPT \ll T$ and when $B \ll \sqrt{T}$. Thus, we expect to improve over this algorithm in this regime, while having similar performance in the large budget regime. The second baseline is the naive $UCB$ algorithm that ignores the budget constraints and maximizes the rewards. The third algorithm is the \emph{The LuekerLearn Algorithm} from \cite{tran-tranh14} adapted to multiple resources and valuation function in the objective. In particular, we run $m$ different copies of this algorithm, one for each platform. At each time-step, we divide the total remaining budget uniformly across the $m$ different platforms. In other words, if the remaining budget at time $t$ is $B_t$, each instance of the LuekerLearn algorithm will receive the remaining budget as $B_t/m$. Thus, after each time-step, each instance of the algorithm has a \emph{synchronization} where unused budget from one platform can be \emph{transferred} to the other platforms. The full algorithm is described in the Appendix. We would like to emphasize that the baselines perform almost as good as the algorithm proposed in this paper when $m$ is small (see Fig.~\ref{fig:varyPlatform}), thus, suggesting that the modified LuekerLearn algorithm is a non-trivial baseline.
 
 \paragraph{Results.} 
 	Fig.~\ref{fig:varyBudget} shows the variation of the total reward collected by each of the three algorithms as a function of budget $B$ for a given advertiser with $T=10^5$ steps. We see that when the total budget becomes larger and a constant fraction of $T$, all algorithms perform nearly well. When the budget becomes smaller, both the $UCB$ algorithm and the modified LuekerLearn algorithm collects lower total reward since they run out of budget very early. Moreover, we see that our algorithm performs better than SemiBwK-RRS in this regime and matches the theory. When the budget increases, both our algorithm and SemiBwK-RRS have comparable performance. In Fig.~\ref{fig:varyPlatform}, we study the effect of the number of platforms on the total reward. We randomly choose a subset $m \subseteq [10]$ and run the three algorithms (with $B=10^3$ and $T=10^5$). We see that with fewer platforms, the difference between the modified LuekerLearn algorithm and our algorithm vanishes. Addtionally, we also see that our algorithm performs slightly better than SemiBwK-RRS when the number of platforms are large. This is unexplained by theory but seems to hold empirically. In Fig.~\ref{fig:stoppingTime}, we look at the average stopping time of the various algorithms. We can see that our algorithm (also SemiBwK-RRS which we omit for clarity) depletes the budget uniformly and runs till the end (\emph{i.e.,} $T=10^5$) while both the baselines deplete their budget within a small fraction of the total time-steps. In advertising platforms, advertisers expect the budget to be used up uniformly over a large period (\emph{e.g.,} a day). Thus, algorithms that deplete the budget very quickly are not desirable, even if they end up collecting larger reward. This is another feature of our algorithm that makes it useful for practice.

\section{Conclusion}
	
In this paper we presented algorithms for budget management in online advertising across $m$ bidding platforms.  We modeled the problem as  Stochastic Bandits with Knapsack with an $m$-dimensional bidding vector. We designed an algorithm for bidding in discrete and continuous actions spaces and proved a mathematical bound on it regret. We also sketched a lower-bound to show that this is optimal. Finally, we used real-world datasets to show extensive empirical evaluation and compared against other competitive heuristics. From a practical stand-point, we believe our work can open directions in modeling for bidding in auctions. First, it is a challenging task to extend to settings where the auction results have correlations across different time-steps. Second, we expect the study of pacing strategies for multi platform advertisement to provide challenging problems for practical and theoretical investigations even beyond the multi-armed bandit setting. 

\onecolumn
 \begin{figure}
\centering
\begin{subfigure}{.3\textwidth}
  \centering
  \includegraphics[width=\linewidth]{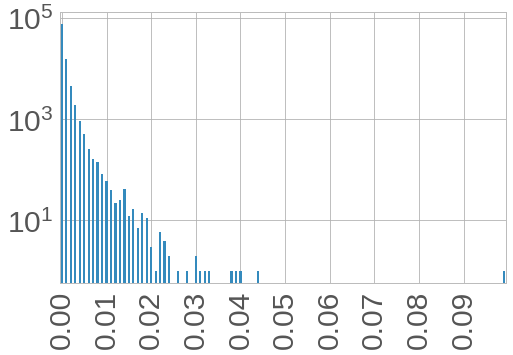}
\end{subfigure}%
\begin{subfigure}{.3\textwidth}
  \centering
  \includegraphics[width=\linewidth]{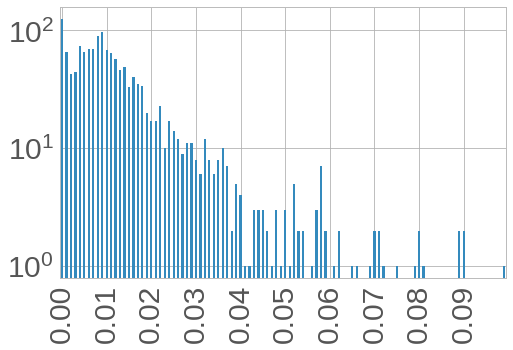}
\end{subfigure}
\begin{subfigure}{.3\textwidth}
  \centering
  \includegraphics[width=\linewidth]{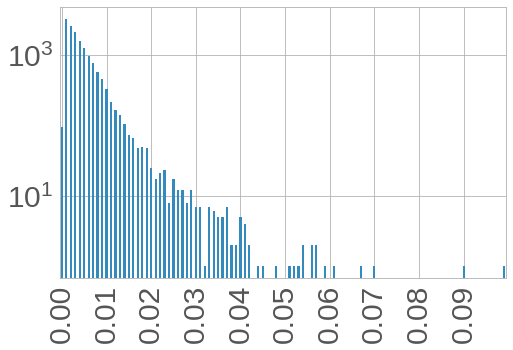}
\end{subfigure}
\caption{Representative price distribution for three platforms and one advertiser}
\label{fig:representativeDistributions}
\end{figure}

\begin{figure}
\centering
\begin{subfigure}{.3\textwidth}
  \centering
  \includegraphics[width=\linewidth]{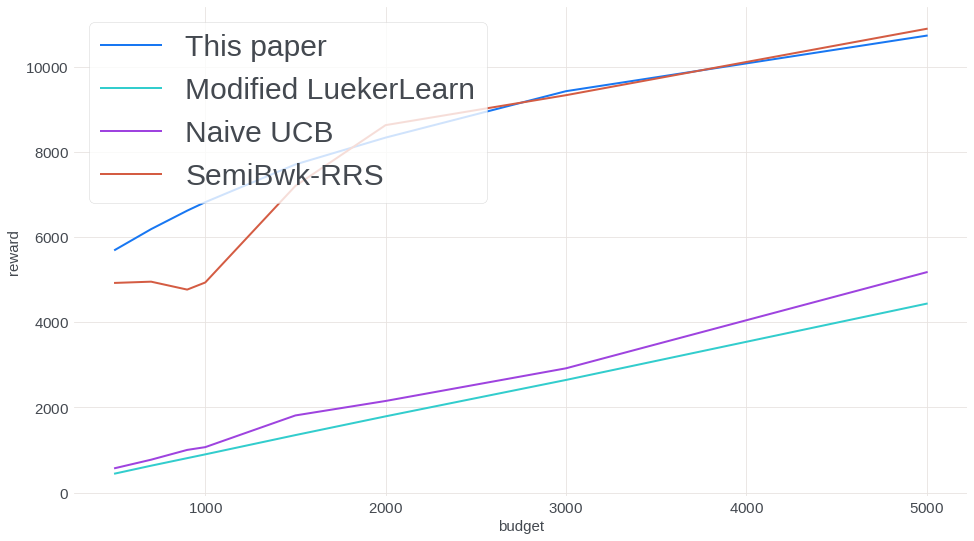}
\end{subfigure}%
\begin{subfigure}{.3\textwidth}
  \centering
  \includegraphics[width=\linewidth]{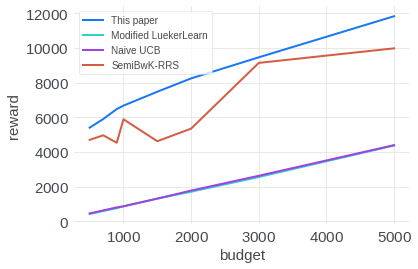}
\end{subfigure}
\begin{subfigure}{.3\textwidth}
  \centering
  \includegraphics[width=\linewidth]{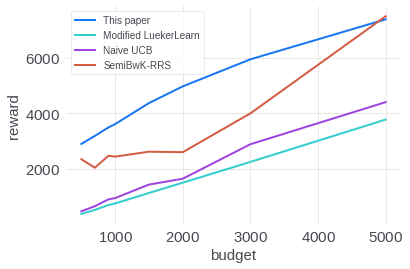}
\end{subfigure}

\caption{Total reward obtained as a function of budget}
\label{fig:varyBudget}
\end{figure}

\begin{figure}[!h]
\centering
\begin{subfigure}{.3\textwidth}
  \centering
  \includegraphics[width=\linewidth]{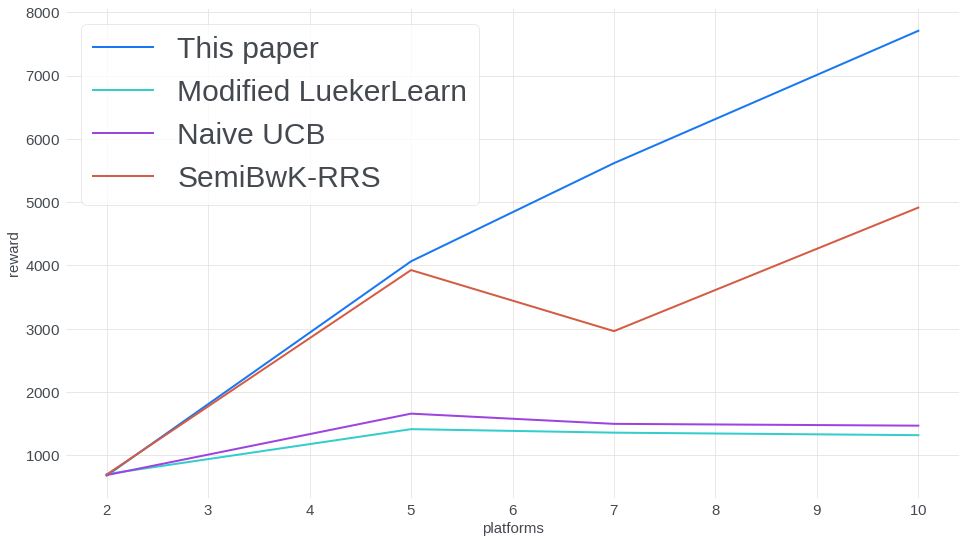}
\end{subfigure}%
\begin{subfigure}{.3\textwidth}
  \centering
  \includegraphics[width=\linewidth]{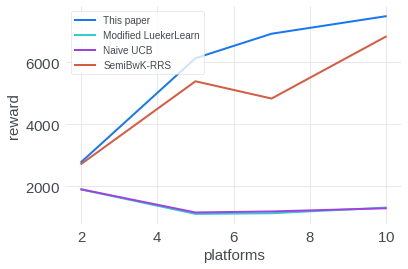}
\end{subfigure}
\begin{subfigure}{.3\textwidth}
  \centering
  \includegraphics[width=\linewidth]{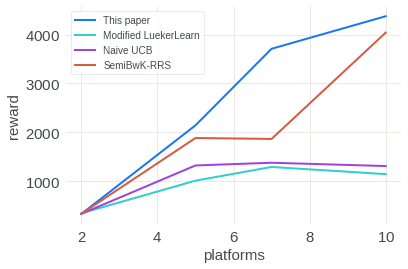}
\end{subfigure}

\caption{Total reward obtained as a function of number of platforms}
\label{fig:varyPlatform}
\end{figure}

\begin{figure}[!h]
\centering
\begin{subfigure}{.3\textwidth}
  \centering
  \includegraphics[width=\linewidth]{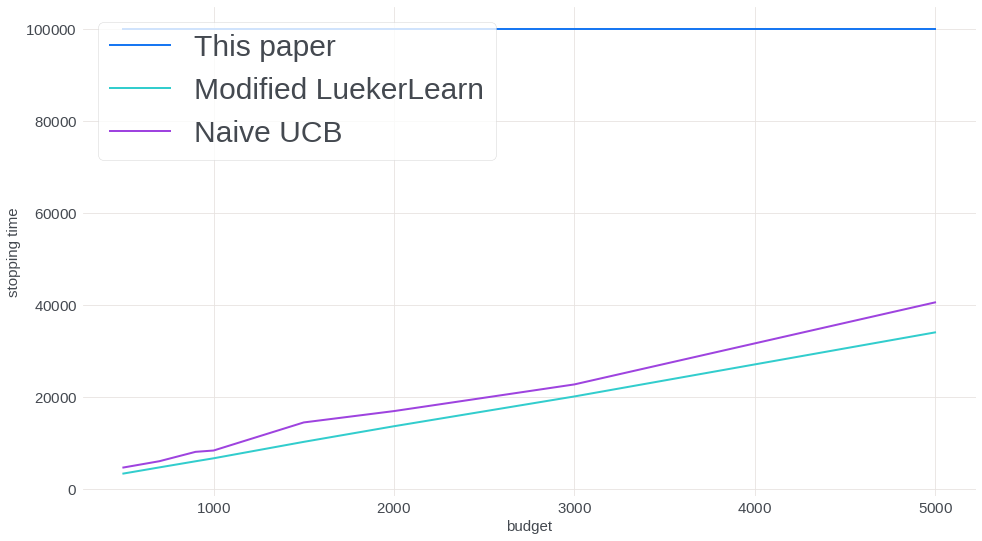}
\end{subfigure}%
\begin{subfigure}{.3\textwidth}
  \centering
  \includegraphics[width=\linewidth]{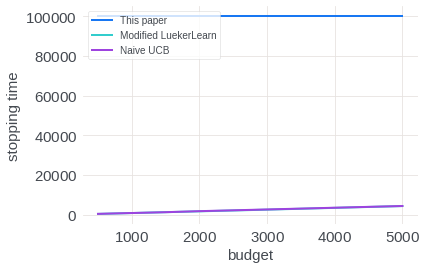}
\end{subfigure}
\begin{subfigure}{.3\textwidth}
  \centering
  \includegraphics[width=\linewidth]{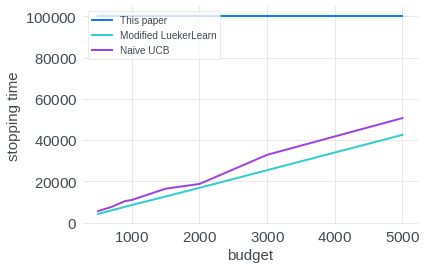}
\end{subfigure}

\caption{Average stopping time of the three algorithms as a function of budget}
\label{fig:stoppingTime}
\end{figure}

\bibliographystyle{plain}
\bibliography{references}

\appendix 

\section{Modified LuekerLearn algorithm}

In this section, we formally describe the modified LuekerLearn algorithm we employ as a baseline in the experimental section. The unbiased maximum likelihood estimate for censored data is given by the Zeng's estimator in Steps 1-5 of \cite{zeng2004estimating}, which was first used in \cite{tran-tranh14}. A simpler estimate, which is not unbiased, is the famous Kaplan-Meier estimator \cite{kaplan1958nonparametric} which was used in \cite{amin2012budget}. After implementing both, we found that the performance was similar and thus, throughout the experimental section we use the simpler Kaplan-Meier estimator. For our purposes, the Kaplan-Meier estimator is defined as follows. Let $D(t', b, p)$ denote the number of time-steps until $t'$ such that bidding $b$ on platform $p$ did not result in a click (\emph{i.e.,} $b < p$) and let $N(t', b, p)$ denote the number of times until $t'$ we bid $b$ on platform $p$. Then, the estimate $\hat{p}_t(b, p)$ is defined as 
	\[
		1-\prod_{t'=1}^{t-1} \left( 1 - \frac{D(t', b, p)}{N(t', b, p)} \right).
	\]

\begin{algorithm}[h]
\caption{Modified LuekerLearn algorithm}
\begin{algorithmic}
	
	\STATE Initialize $B_1 = B$, $\hat{p}_1(i, b) = 1$ for all platforms $i \in [m]$ and all bids $b \in [0, 1]$.
	\FOR{$t=1, \ldots, \tau $ (i.e., until resource budget is exhausted)}
		\STATE{ For each platform $p \in [m]$ play the bid that maximizes the following 
		\begin{eqnarray} 
		j_p^* = \argmax_{b \in [0, 1]} &  b \\
		{\rm s.t.} ~ & \sum_{ 0 \leq \sigma \leq b} \hat{p}_t(p, \sigma) \sigma \leq \frac{B_t}{m(T-t+1)}
		\end{eqnarray} }
		\STATE{Play arm $x_{\bm b_t}$ with bid $\bm b_t =  \{b_{p,j_p^*} \}_{p\in [m]}$}
		\STATE{Update $B_t$ to be the residual budget}
		\STATE{Update $\hat{p}$ for each bid and platform using the unbiased maximum likelihood estimate for censored data.}
		\ENDFOR	
	
\end{algorithmic}

\label{alg:modifiedLuekerLearn}

\end{algorithm}

\end{document}